\documentclass[11pt]{article}
\usepackage[margin=1in]{geometry}
\usepackage{times}
\usepackage[
backend=biber,
style=alphabetic,
sorting=nyt,
maxnames=7,
maxalphanames=7,
backref=true,
url=false
]{biblatex}
\usepackage{tikz}
\usepackage{scalerel}
\usepackage{pict2e}
\usepackage{setspace}
\usepackage{fancyhdr}
\usepackage{tkz-euclide}
\usetikzlibrary{calc}
\usetikzlibrary{patterns,arrows.meta}
\usetikzlibrary{shadows}
\usetikzlibrary{external}
 
\usepackage{pgfplots}
\pgfplotsset{compat=newest}
\usepgfplotslibrary{statistics}
\usepgfplotslibrary{fillbetween}

\usepackage{xcolor}

\usepackage{nicefrac}
\usepackage[colorlinks=true,citecolor=blue,linkcolor=blue]{hyperref}

\hypersetup{
    colorlinks,
    linkcolor={red!50!black},
    citecolor={blue!50!black},
    urlcolor={blue!80!black}
}

\usepackage[]{amsmath,amssymb,amsfonts,latexsym,amsthm,enumerate,fullpage,xcolor,bbm,mathtools}
\usepackage{wrapfig}
\usepackage{float}
\usepackage[ruled]{algorithm2e}

\usepackage{enumitem}
\usepackage[nameinlink, noabbrev, capitalize]{cleveref}
\usepackage{comment}
\usepackage{nicefrac}
\crefname{prop}{Proposition}{Propositions}
\crefname{ineq}{inequality}{inequalities}
\creflabelformat{ineq}{#2(#1)#3}
\usepackage{todonotes}
\usepackage{multicol}
\allowdisplaybreaks[4]

\newtheorem{globtheorem}{Theorem}

\newtheorem{counter}{Counter}[section]
\newtheorem{theorem}[counter]{Theorem}

\newtheorem{lemma}[counter]{Lemma}

\newtheorem{claim}[counter]{Claim}
\newtheorem{fact}[counter]{Fact}

\newtheorem{corollary}[counter]{Corollary}

\newtheorem{definition}[counter]{Definition}

\newcommand{\N}{\mathbb{N}}

\newcommand{\U}{\mathbf{U}}

\newcommand{\X}{\mathbf{X}}

\newcommand{\Y}{\mathbf{Y}}

\newcommand{\A}{\mathbf{A}}

\renewcommand{\multicitedelim}{\addsemicolon\space}

\newcommand\numberthis{\addtocounter{equation}{1}\tag{\theequation}}

\newcommand{\zo}{\{0,1\}}

\newcommand{\abs}[1]{\left\lvert #1 \right\rvert}
\newcommand{\norm}[1]{\left\Vert #1 \right\Vert}

\usepackage{xspace}

\newcommand{\eps}{\varepsilon}

\DeclareMathOperator{\I}{\mathbf{I}}
\DeclareMathOperator{\J}{\mathbf{J}}
\let\P\relax
\DeclareMathOperator{\P}{\mathbf{P}}

\DeclareMathOperator*{\E}{\mathbb{E}}

\usepackage{subfiles}
\newcommand{\dobib}{
    \printbibliography
}

\begin{document}
\renewcommand{\dobib}{}
\renewcommand*{\multicitedelim}{\addcomma\space}
 
\title{A Resolution of Friedgut's Conjecture on Influential Coalitions}

\author{Eshan Chattopadhyay\thanks{This work was supported in part by NSF Award CCF-2514586.}\\ Cornell University\\ 
\texttt{eshan@cs.cornell.edu}  
\and Mohit Gurumukhani\footnotemark[1] \\ Cornell University\\ \texttt{mgurumuk@cs.cornell.edu} 
}

\date{}

\maketitle

\begin{abstract}
We prove that, for every constant $\varepsilon>0$ and every function $f:\Sigma^n\to\{0, 1\}$, there is a coalition of $O(n/\sqrt{\log n})$ coordinates and a target output $b\in\{0, 1\}$ such that, after the remaining coordinates are sampled uniformly and independently, the coalition can choose its values to make the output equal to $b$ with probability at least $1-\varepsilon$. The bound is independent of the alphabet size and also holds for monotone Boolean functions on $[0,1]^n$, resolving a conjecture of Friedgut (Combinatorics, Probability and Computing, 2004). Unlike the Boolean cube setting, where Kahn, Kalai, and Linial (FOCS, 1988) give a coalition bound of $O(n/\log n)$, no sublinear bound independent of the alphabet size was previously known.

In collective coin flipping, our result gives the first sublinear bound on the number of bad players needed to force a fixed output with probability at least $1-\varepsilon$ in any one-round protocol with independent uniform messages, regardless of the message length.

A key ingredient in our proof is an encoding that lets us relate the influence of a function on a product space to the $p$-biased influence of the encoded function. We then rely on a structure theorem of Hatami (Annals of Mathematics, 2012) for functions with small $p$-biased influence to bias the encoded function.
\end{abstract}

\section{Introduction}
The Kahn--Kalai--Linial (KKL) theorem~\cite{KKL88influence}
is a fundamental result in the analysis of Boolean functions.
It states that every balanced Boolean function
$f:\zo^n\to\zo$ has a coordinate with influence at least
$\Omega(\log n/n)$, where the influence of a coordinate is the
probability that changing its value changes the output.
An important consequence is that a small coalition of coordinates
can almost determine the output. More precisely, for every constant $\eps>0$ and any target output $b\in\zo$, there is
a set of $O(n/\log n)$ coordinates such that, after the remaining
coordinates are sampled uniformly and independently, the
coalition can choose its values to make the output of $f$ equal
to $b$ with probability at least $1-\eps$.

The same question can be asked for functions
$f:\Sigma^n\to\zo$ over an arbitrary finite alphabet $\Sigma$.
Does a sublinear coalition still suffice, with a bound
independent of $|\Sigma|$?
In this setting, the influence of a coordinate is the
probability that, after the remaining coordinates are sampled
uniformly and independently, its value can be chosen to make
the output either $0$ or $1$.
Bourgain, Kahn, Kalai, Katznelson, and Linial
(BKKKL)~\cite{BKKKL92product} extended the KKL theorem to
these product spaces. Their theorem states that every balanced
function has a coordinate with influence at least
$\Omega(\log n/n)$, with no dependence on the alphabet size.
More generally, their theorem also applies when each coordinate
is uniform over $[0,1]$ and
$f:[0,1]^n\to\zo$.\footnote{In the continuous setting, we consider Boolean functions
on $[0,1]^n$ that are either monotone or Borel measurable,
with probabilities taken under Lebesgue measure.}

However, the consequence for coalitions no longer holds:
a small coalition need not be able to force an arbitrary
target output with probability close to one.
To see this, consider a finite-alphabet version of the example
discussed by Friedgut~\cite{friedgut04influences}.
Set $\Sigma=\{0,1,\ldots,n-1\}$ and define
$f:\Sigma^n\to\zo$ by
\[
  f(x)=1 \quad\Longleftrightarrow\quad
  x_i\ne 0 \text{ for every } i\in[n].
\]
The probability that $f(x)=1$ is $(1-1/n)^n$, which tends
to $1/e$. Every coordinate has influence $(1-1/n)^{n-1}$:
it can determine the output exactly when all the other
coordinates are nonzero.
Nevertheless, a coalition of $s=o(n)$ coordinates can make
the output equal to $1$ with probability at most
$(1-1/n)^{n-s}$, which still tends to $1/e$.
This is because every coordinate outside the coalition must
be nonzero, regardless of how the coalition chooses its values.
In contrast, a single coordinate can force the output to be
$0$ by setting its value to $0$.

This example illustrates why repeatedly applying BKKKL
does not give the same coalition guarantee.
For this discussion, it suffices to consider monotone
functions (see \cref{claim:monotonization}).
For a monotone Boolean function on $\zo^n$, fixing a
coordinate to $1$ gives a function on the remaining
coordinates whose expectation is larger by half the
influence of that coordinate. Fixing it to $0$ instead
gives a function whose expectation is smaller by the
same amount. We can then apply KKL to the restricted
function and repeat, making progress towards the same
target output at each step.
For a monotone Boolean function over a larger ordered
alphabet, fixing a coordinate to its largest value
increases the expectation of the restricted function,
while fixing it to its smallest value decreases it.
The sum of the increase and the decrease still equals
the influence, but one may be much smaller than the other.
In the example above, fixing a coordinate to a nonzero
value gives a function whose expectation exceeds the
original expectation by only $(1-1/n)^{n-1}/n$, despite
the coordinate having constant influence.
Thus, BKKKL may give an influential coordinate for which
no restriction makes substantial progress towards
our chosen target. Choosing whichever target gives the
larger change in expectation does not resolve the issue:
after restricting a coordinate, the next step may move
the expectation in the opposite direction.

The example leaves open whether a small coalition can always
force at least one of the two outputs.
Friedgut~\cite{friedgut04influences} conjectured that this
is indeed the case, formulating the question for Boolean
functions on $[0,1]^n$.
In the finite-alphabet setting, the question asks whether,
for every constant $\eps>0$ and every function
$f:\Sigma^n\to\zo$, there is a set of $o(n)$ coordinates,
with a bound independent of $|\Sigma|$, and a target output
$b\in\zo$ such that, after the remaining coordinates are
sampled uniformly and independently, the coalition can
choose its values to make the output equal to $b$ with
probability at least $1-\eps$. Towards this, Filmus, Hambardzumyan, Hatami, Hatami, and
Zuckerman~\cite{FHHHZ19} showed that, under any product
distribution on $\zo^n$, a coalition of
$O(n\log\log n/\log n)$ coordinates can force some fixed
output with probability at least $1-\eps$, for any
constant $\eps>0$. This allows arbitrary coordinate biases,
but encoding larger alphabets as bits introduces a
dependence on the message length.

We resolve Friedgut's conjecture, obtaining a coalition of
size $O(n/\sqrt{\log n})$ in both the finite-alphabet and
continuous settings.  

\begin{globtheorem}[Informal version of \cref{thm: corrupt finite alphabet,thm:main_influence_uniform_measure}]\label{informal_thm:intro_main}
For every constant $\eps>0$, every finite alphabet $\Sigma$,
and every function $f:\Sigma^n\to\zo$, there is a coalition
of $O(n/\sqrt{\log n})$ coordinates and a target output
$b\in\zo$ such that, after the remaining coordinates are
sampled uniformly and independently, the coalition can
choose its values to make the output equal to $b$ with
probability at least $1-\eps$.
The implicit constant depends only on $\eps$.
The same guarantee holds for monotone functions as well as Borel measurable Boolean functions on $[0,1]^n$.
\end{globtheorem}

Recall that, on the Boolean cube, KKL gives a coalition
of size $O(n/\log n)$. In the other direction, Ajtai and
Linial~\cite{AL93resilient} showed the existence of almost-balanced functions $f:\zo^n\to\zo$ 
that are resilient to coalitions of size
$\Omega(n/(\log n)^2)$. More precisely, for every constant
$\delta>0$, there is a constant $c>0$ such that every
coalition of at most $cn/(\log n)^2$ coordinates can change
the probability of either output by at most $\delta$.
This leaves a gap between the lower bound of
$\Omega(n/(\log n)^2)$, which already holds on the Boolean
cube, and our upper bound of $O(n/\sqrt{\log n})$.

\paragraph{Connection to collective coin flipping.}
The study of influential coalitions originated in collective
coin flipping, introduced by Ben-Or and Linial~\cite{BL85coin},
where players try to agree on a random bit despite adversarial
faults. This is a fundamental problem in fault-tolerant distributed computing and has been extensively studied (see the survey of Dodis~\cite{dodis2006fault}). 
In the full information model, players broadcast messages
using independent private randomness. An unbounded adversary
controls a fixed set of bad players and chooses their messages
after observing the good players' messages in the current round.
For one-round protocols with uniform messages from $\Sigma$,
this is exactly the setting discussed above: each coordinate
is a player's message, the function $f:\Sigma^n\to\zo$
determines the output, and controlling a player allows the
adversary to choose that coordinate after observing all
coordinates outside the coalition.

Thus, for single-bit messages, KKL shows that
$O(n/\log n)$ bad players suffice to force a fixed output
with probability at least $1-\eps$, for any constant
$\eps>0$. Our result gives a bound of
$O(n/\sqrt{\log n})$ bad players, regardless of the
message length.

\paragraph{Connections to extractor theory.}
In extractor theory, functions with large alphabets and their pseudorandom properties under adversarial settings (as here) have been studied as constructing extractors and condensers for non-oblivious symbol-fixing (NOSF) sources. Under that terminology, our results show that it is impossible to extract uniform random bits from NOSF sources in the presence of $O\left(\frac{n}{\sqrt{\log n}}\right)$ bad blocks.
Such constructions have been extremely useful in constructions of extractors for independent sources \cite{CZ19, Li23}.
Related questions regarding extractors and condensers for NOSF sources have also been studied in extractor theory \cite{KN23, CGR24condenser}, with geometric interpretations of these studied in \cite{KS24shadow}.

\subsection{Proof Overview}
 We explain the main ideas in the proof of
\cref{informal_thm:intro_main} for monotone functions
$f:[r]^n\to\zo$, where $[r]=\{1,\ldots,r\}$.
Throughout this overview, we fix a constant error $\eps>0$.
As discussed above, repeatedly choosing an influential coordinate
need not make progress towards the same output.
Our approach is to decrease the expectation towards $0$
for as long as we can make sufficient progress, and then
show that a small coalition can force $1$ if this process stops.

For a monotone function $h:[r]^n\to\zo$, let $\J_j^0(h)$
denote the increase in the probability of output $0$ when
coordinate $j$ can be chosen after observing the remaining
coordinates, which are sampled uniformly and independently.
By monotonicity, this is exactly the decrease in expectation
obtained by fixing coordinate $j$ to its smallest value.
Set $\J^0(h)=\sum_j\J_j^0(h)$.
(See \cref{subsec:influence_defs} for formal definitions.)
Our key result is that, if $\E[h]\ge\eps$ and
$\J^0(h)\le\tau$ for $\tau > 0$, then a coalition of
$\exp(O(\tau^2))$ coordinates can force $1$ with probability
at least $1-\eps$. The implicit constant depends only on $\eps$.
This is formally stated as
\cref{lem: small 0 influence implies small coalition bias 1}.
Thus, small total influence towards $0$ gives a small
coalition that can force the opposite output.

Before proving the above key result, we show how it directly
implies our theorem. Starting with $f$, we repeatedly fix
coordinates to their smallest value to decrease the expectation.
Set $\delta=c\sqrt{\log n}/n$ for a sufficiently small
constant $c>0$. As long as the expectation exceeds $\eps$
and some coordinate has $\J_j^0\ge\delta$, we fix that
coordinate to its smallest value, namely $1$.
Each step decreases the expectation by at least $\delta$,
so there are at most $1/\delta=O(n/\sqrt{\log n})$ such steps.
If the expectation falls to at most $\eps$, then the coordinates we fixed form a coalition forcing the output $0$ with high probability, as desired.
Otherwise, the remaining function has expectation greater
than $\eps$ and total influence towards $0$ at most
$n\delta=c\sqrt{\log n}$.
Applying \cref{lem: small 0 influence implies small coalition bias 1},
we obtain a coalition of size $\exp(O(c^2\log n))$ that
can force the output to be $1$ with probability at least $1-\eps$.
For sufficiently small $c$, this coalition has size
at most $\sqrt n$.
By monotonicity, this coalition biases the output towards $1$ even without
controlling the previously fixed coordinates.
Thus, in either case, we obtain a coalition of size $O(n/\sqrt{\log n})$.

We next explain how to prove
\cref{lem: small 0 influence implies small coalition bias 1},
the central ingredient in our argument.
Our starting point is a theorem of Hatami~\cite{hatami12structure}
that gives a small coalition whenever the total $p$-biased
influence is small.
For a Boolean function on independent bits that each equal $1$
with probability $p$, the $p$-biased influence of a bit is
the probability that independently resampling the bit from the same distribution changes the output.
The total $p$-biased influence is the sum of these probabilities
over all bits.
(See \cref{def:resampling_influence} for a formal definition.)
The consequence of Hatami's theorem that we use states
that for a monotone function on independent $p$-biased bits (where $0<p\le 1/2$), with expectation at least $\eps$ and
total $p$-biased influence at most $\tau$, a coalition of
$\exp(O(\tau^2))$ bits can force the output to be $1$ with probability at least $1-\eps$.
The bound depends only on $\tau$ and $\eps$, even if there are arbitrarily many bits and $p$ is arbitrarily small. (We refer the reader to \cref{thm: hatami pseudo junta} for the precise statement.)

This gives a route to our lemma. Recall that we are given
a monotone function $h:[r]^n\to\zo$ with $\E[h]\ge\eps$
and $\J^0(h)\le\tau$.
We encode each coordinate of $h$ by independent $p$-biased
bits so that the resulting function has total $p$-biased
influence at most $O(\tau)$.
We must also approximately preserve the expectation and
ensure that a coalition of bits biasing the encoded function gives a coalition of at
most as many original coordinates biasing $h$. Applying Hatami's theorem to the encoded function would then give the desired bound.

To obtain such an encoding, we represent each coordinate
by independent bits that equal $1$ with probability $p$.
To do this, each bit is assigned a value in $[r]$, and the encoding outputs the largest value assigned to a bit that equals
$1$ (or $1$ if all bits are $0$).
With suitable repetitions of the assigned values and
sufficiently small $p$, the output approximates the uniform
distribution on $[r]$.
The important property is that every (non-constant) threshold function over $[r]$ corresponds to an OR function over a subset of the $p$-biased bits.
Since fixing all other coordinates of a monotone $h$ leaves
a threshold or a constant function, this lets us bound the
total $p$-biased influence of the encoded function by
$2\J^0(h)$, up to an arbitrarily small approximation error.
This is proved in \cref{lem: approx product p biased}.
Now applying \cref{thm: hatami pseudo junta} gives the desired
coalition of bits. These bits belong to at most as many original
coordinates, so controlling those coordinates gives a
coalition for $h$ with no dependence on $r$, as desired.

\paragraph{Organization.} We introduce necessary background in \cref{sec:prelim}. In \cref{sec:coalition_bound}, we prove our main result that bounds the size of an influential coalition, using a reduction to the biased Boolean cube. We present and analyze this reduction in \cref{sec:biased_space}. 
\section{Preliminaries}\label{sec:prelim}

\paragraph{Notation.}
All logarithms, unless specified otherwise, have base $2$.
For $b\in \zo$, we write $\overline{b}$ to mean $1-b$.
For a string $x\in\Sigma^n$ and a set $S\subseteq[n]$,
let $x_S$ denote its restriction to the coordinates in
$S$, and let $x_{-S}$ denote its restriction to the
remaining coordinates.
We use $\U_n^r$ to denote the uniform distribution over
$[r]^n$. When $r$ is clear from context, we write $\U_n$.
For $x,y$ in either $\zo^n$ or $[r]^n$ or $[0, 1]^n$, we write
$x\le y$ if $x_i\le y_i$ for every $i\in[n]$.
A Boolean function $f$ on either domain is monotone
if $x\le y$ implies $f(x)\le f(y)$.

\subsection{Influence of Boolean Functions}\label{subsec:influence_defs}
We use two notions of influence: Influence
on the biased Boolean cube and influence towards a
fixed output in product spaces.

\paragraph{p-biased influence.}
Let $\mu_p^n$ denote the distribution on $\zo^n$ obtained
by independently setting each bit to $1$ with probability
$p$. When $n$ is clear from context, we will just write
$\mu_p$.

We use the following notion of influence with
respect to the $p$-biased distribution, as in
Hatami~\cite{hatami12structure}.

\begin{definition}[$p$-biased influence]
\label{def:resampling_influence}
For any $n\in\N$, $p\in[0,1]$, and function
$f:\zo^n\to\zo$, we define the $p$-biased influence of
variable $i$ on $f$ as:
\[
    \I^p_i(f)
    = \Pr_{x,y\sim\mu_p}\big[
    f(x_1,\dots,x_{i-1},x_i,x_{i+1},\dots,x_n)
    \ne
    f(x_1,\dots,x_{i-1},y_i,x_{i+1},\dots,x_n)
    \big].
\]
With this, we define the total $p$-biased influence of
$f$ as:
\[
    \I^p(f)=\sum_{i=1}^n\I^p_i(f).
\]
\end{definition}

\paragraph{Influence towards a fixed output.}
We next define the influence towards a direction.
Before that, we need notation for the probability that
a coalition can force a given output.

Throughout, let $\Sigma$ be a finite set, let $n\in\N$, let $f:\Sigma^n\to\zo$ be a
 function, and let $\sigma$ be an arbitrary
distribution on $\Sigma$. Let $\sigma^n$ be the product distribution
obtained by sampling each coordinate independently from $\sigma$.

For $i\in[n]$ and $b\in\zo$, we define the following quantity: 
\[
    \P^b_i(f,\sigma) = \Pr_{x\sim\sigma^n}\big[ \exists j\in \Sigma: f(x_1,\dots,x_{i-1},j,x_{i+1},\dots,x_n)=b \big].
\]
We also extend this notion to sets.
For $S\subseteq[n]$ and $b\in\zo$, define
\[
    \P^b_S(f,\sigma)
    = \Pr_{x\sim\sigma^n}\big[
    \exists j\in\Sigma^S: f(j,x_{-S})=b
    \big].
\]
Often we do not mention the distribution $\sigma$, in which case the underlying distribution is uniform over the domain. We use the same definitions for monotone or Borel measurable functions $f:[0,1]^n\to\zo$, with each coordinate sampled uniformly
from $[0,1]$ and coalition assignments chosen from $[0,1]^S$. These probabilities are well defined when $f$ is Borel measurable: coalition-success events are projections of Borel sets, and hence are analytic and Lebesgue measurable.\footnote{This is not true for functions $f$ that are merely Lebesgue measurable. Hence, all our statements are stated only for monotone $f$ or Borel measurable $f$.} When $f$ is monotone, the coalition can force output $b$ precisely when setting every coordinate in $S$ to $b$
gives output $b$. The resulting restricted function is monotone and hence Lebesgue measurable.

Using these, define the influence towards direction
$b\in\zo$ of coordinate $i\in[n]$ as:
\[
    \J^b_i(f,\sigma)
    = \P^b_i(f,\sigma)
    - \Pr_{x\sim\sigma^n}[f(x)=b].
\]
We finally define
$\J^1(f,\sigma)=\sum_i\J^1_i(f,\sigma)$ and
$\J^0(f,\sigma)=\sum_i\J^0_i(f,\sigma)$.
As above, when we omit $\sigma$, the underlying distribution will be the uniform distribution.

\subsection{Monotonizing and Discretizing Boolean functions}

We present results that let us without loss of generality assume that our given function is monotone and is over finite alphabet.

\paragraph{Monotonizing Boolean Functions}
When proving lower bounds on the influence of coalitions, it suffices
to consider monotone Boolean functions. This follows from the following
monotonization lemma, proved using a sifting argument similar to those
in \cite{BL85coin, KKL88influence, BKKKL92product}.
\begin{claim}\label{claim:monotonization}
For every Borel measurable function $f: [0,1]^n \to \zo$, there exists a monotone 
function $g: [0,1]^n \to \zo$ such that $\E[g]=\E[f]$ and,
for all $S\subseteq [n]$ and $b\in \zo$, it holds that
$\P^b_S(g) \le \P^b_S(f)$.
Furthermore, if $f: [r]^n \to \zo$, then such a monotone
function $g: [r]^n \to \zo$ exists with these properties.
\end{claim}

\begin{proof}
This follows by a standard sifting argument.
Write $D=[0,1]$ or $D=[r]$, as appropriate, and let $\mu$
denote the uniform probability measure on $D$.
For each $i\in [n]$ and assignment $x_{-i}\in D^{n-1}$,
let $\ell_{i,x}(t)=f(x_{-i},t)$.
We replace this function by a monotone function
$m_{i,x}:D\to\zo$ with the same expectation, as follows.

If $\E[\ell_{i,x}] = b$ for $b\in \zo$, then we define $m_{i, x} \equiv b$.
Otherwise, when $D=[0,1]$, define $m_{i,x}(t)=1$ if and only if
$t\ge 1-\E[\ell_{i,x}]$.
When $D=[r]$, set
$k=\abs{\{t\in[r]:\ell_{i,x}(t)=1\}}$ and define
$m_{i,x}(t)=1$ if and only if $t>r-k$.
Thus, in the discrete case, the ones occupy precisely the
largest $k$ elements of $[r]$.

We monotonize $f$ coordinate by coordinate, with the operation
in coordinate $i$ replacing every $\ell_{i,x}$ by $m_{i,x}$.
Each operation preserves the expectation.
It also preserves monotonicity in every previously treated
coordinate: if two assignments to the other coordinates are
ordered in such a coordinate, then their original functions
are pointwise ordered. Their expectations are therefore ordered,
and the replacements remain pointwise ordered, including when
one of the original functions has expectation $0$ or $1$.
Consequently, the final function $g$ is monotone and satisfies
$\E[g]=\E[f]$.
We argue each step preserves Borel measurability, so that the operations can be iterated, below.

It remains to show that a single operation, taking $f$ to $f'$
by monotonizing coordinate $i$, satisfies
$\P^b_S(f')\le \P^b_S(f)$ for every $S\subseteq[n]$ and $b\in\zo$ and that $f'$ is Borel measurable.
First, suppose that $i\in S$.
The replacement does not introduce an output value that was
absent from the original one-variable function.
Indeed, if the expectation of this restricted function were $e\in \zo$, then it must attain the value $e$ over a measure $1$ set, and hence changing this to constant $e$ does not introduce any new value; otherwise, the output values available to the functions are $\zo$ and are left unchanged.
Thus, every assignment outside $S$ that allows the coalition
to force $b$ under $f'$ also allows it to force $b$ under $f$.

Now suppose that $i\notin S$.
Fix an assignment $\alpha$ to the coordinates in
$[n]\setminus(S\cup\{i\})$.
For each $\beta\in D^S$, define
\[
    C_\beta=\{t\in D:f(t,\alpha,\beta)=b\},
    \qquad
    C'_\beta=\{t\in D:f'(t,\alpha,\beta)=b\}.
\]
Conditioned on $\alpha$, the probability that the coalition
can force $b$ under $f$ is
$\mu(\bigcup_\beta C_\beta)$.
The replacement preserves
$\mu(C'_\beta)=\mu(C_\beta)$.
Moreover, the sets $C'_\beta$ are all upper intervals when $b=1$
and all lower intervals when $b=0$, and hence are nested.
For either choice of $D$, this gives
\[
    \mu\left(\bigcup_\beta C'_\beta\right)
    =\sup_\beta\mu(C'_\beta)
    =\sup_\beta\mu(C_\beta)
    \le \mu\left(\bigcup_\beta C_\beta\right).
\]
Averaging over $\alpha$, we obtain
$\P^b_S(f')\le \P^b_S(f)$.
Applying this at every step proves the claim in both cases.

Lastly, we show that $f'$ thus obtained is Borel measurable. First, the map $x_{-i}\mapsto\E[\ell_{i,x}]$ is Borel measurable by the Fubini--Tonelli theorem.
Second, $m_{i, x}$ depends on $\ell_{i, x}$ only through $\theta = \E[\ell_{i, x}]$. Letting $m_{\theta}$ be the corresponding function, we see that the map $(\theta,t)\mapsto m_{\theta}(t)$ is the indicator of a Borel subset of $[0, 1]\times D$. As $f'(x) = m_{\E[\ell_{i, x}]}(x_i)$, a composition of these Borel maps, it is Borel measurable as claimed.
\end{proof}

\paragraph{Discretizing Boolean Functions}
We next show how to discretize a monotone function while
approximately preserving its expectation and the probability
that any coalition can force either output.
\begin{claim}\label{claim:discretization}
For all $n, r\in \N$, with $r \ge 2$, the following holds.
For any monotone function $f: [0, 1]^n \to \zo$, there exists a monotone function $g: [r]^n \to \zo$ such that $\abs{\E[f] - \E[g]}\le \frac{n}{r}$, and for any $S\subseteq [n]$ and $b\in \zo$, we have that $\abs{\P^b_S(g) - \P^b_S(f)} \le \frac{n}{r}$.
\end{claim}

\begin{proof}
We define $g$ as follows:
\[
g(y) = f\left(\frac{y_1-1}{r-1}, \dots, \frac{y_n-1}{r-1}\right).
\]
As $f$ is monotone, we immediately have that $g$ too is monotone.
We then proceed by an appropriate hybrid argument, discretizing one coordinate at a time.
When discretizing a single coordinate, by monotonicity of $f$, the change in expectation is at most $\frac{1}{r}$. This can be seen as follows: on fixing all other coordinates, the restricted function is either constant or a threshold. Thus, its probability under the uniform distribution on $\{0,1/(r-1),\dots,1\}$ differs from
its probability under the uniform distribution on $[0,1]$ by at most $1/r$.
By triangle inequality, we infer that $\abs{\E[f] - \E[g]} \le \frac{n}{r}$ as desired.

We similarly argue regarding $\P^b_S(g)$ and $\P^b_S(f)$ for any $S\subseteq [n]$ and $b\in \zo$.
First, since $f$ and $g$ are monotone and discretization maps values $1$ and $r$ to $0$ and $1$ respectively, we see that discretizing coordinates of $S$ preserves these quantities, and we assume that those coordinates are fixed to $1$ or $r$ for $b = 0$ and $b = 1$ respectively.
Then, for other coordinates, we argue as earlier and discretize coordinate by coordinate, observing that the change in probability is at most $\frac{1}{r}$.
Applying appropriate triangle inequalities then, we obtain that $\abs{\P^b_S(g) - \P^b_S(f)} \le \frac{n - |S|}{r} \le \frac{n}{r}$, as desired.
\end{proof}

\subsection{Basic Probability notions}

For distributions $\X,\Y$ on the same finite set $\Omega$, let $\abs{\X-\Y}=\frac12\sum_{x\in\Omega} \abs{\Pr[\X=x]-\Pr[\Y=x]}$ denote their statistical distance.
Next, let $\norm{\X-\Y}_{\infty} = \max_{w \in \Omega} \abs{\Pr[\X = w] - \Pr[\Y = w]}$ denote the $\ell_{\infty}$ distance between them.

We record the following useful inequality.
\begin{fact}[Data Processing Inequality]
\label{fact: dpi}
Let $\Omega, \Omega'$ be arbitrary finite sets.
Let $\X, \Y$ be distributions on $\Omega$ such that $\abs{\X-\Y}\le \eps$.
Then, for any $f: \Omega \to \Omega'$, we have that $\abs{f(\X) - f(\Y)}\le \eps$.
\end{fact}

From this, we easily obtain the following inequality.
\begin{claim}
\label{claim: data processing expectation}
Let $\X,\Y$ be distributions on a finite set $\Omega$ and let $\eps > 0$ be such that $|\X - \Y| \le \eps$.
Then, for any $f: \Omega\to \zo$, we have that
\[
\abs{\E_{x\sim \X}[f(x)] - \E_{y\sim \Y}[f(y)]} \le \eps.
\]
\end{claim}

\begin{proof}
Fix any such $f, \X, \Y$.
Then, 
\begin{align*}
\abs{\E_{x\sim \X}[f(x)] - \E_{y\sim \Y}[f(y)]}   
& = \abs{f(\X) - f(\Y)}\\
& \le \eps,
\end{align*}
where the last step follows by \cref{fact: dpi}.
\end{proof}

\section{A bound on influential coalition size} \label{sec:coalition_bound}

In this section, we will prove the following main result regarding biasing a function over $[r]^n$.
\begin{theorem}
\label{thm: corrupt finite alphabet}
For all constant $\eps > 0$, there exists a constant $C = C(\eps) > 0$ such that the following holds.
For all $r, n\in \N$ and function $f: [r]^n \to \zo$, there exists $b\in \zo$, and $B\subseteq [n]$ with $\abs{B}\le C\cdot \frac{n}{\sqrt{\log(n)}}$ such that
\[
    \P_B^b(f) \ge 1 - \eps.
\]
\end{theorem}

We obtain the same coalition bound for monotone or Borel measurable Boolean functions on $[0,1]^n$.
\begin{corollary}\label{thm:main_influence_uniform_measure}
For all constant $\eps > 0$, there exists a constant $C = C(\eps) > 0$ such that the following holds.
For any $n\in \N$, consider $[0, 1]^n$ with the uniform measure.
Let $f: [0, 1]^n \to \zo$ be an arbitrary monotone or Borel measurable function. Then, there exists $b\in \zo$ and $B\subseteq [n]$ with $\abs{B}\le C\cdot \frac{n}{\sqrt{\log(n)}}$ such that
\[
\P^b_B(f) \ge 1 - \eps.
\]
\end{corollary}

\begin{proof}
If $f$ is non-monotone (and so by assumption Borel measurable), then we apply \cref{claim:monotonization} and without loss of generality assume that $f$ is a monotone function.
Then, we apply \cref{claim:discretization} with $ r = \max\{2, \lceil 2n/\eps \rceil \} $ to obtain a function $g: [r]^n\to \zo$ such that if there exists $B\subseteq [n]$ and $b\in \zo$ such that $\P^b_B(g) \ge 1-\frac{\eps}{2}$, then $\P^b_B(f) \ge 1-\eps$.
We finally apply \cref{thm: corrupt finite alphabet} to $g$ with error parameter $\frac{\eps}{2}$ and find such $b, B$ to obtain the desired result.
\end{proof}

To prove \cref{thm: corrupt finite alphabet}, we will need the following result that shows that if total influence towards $0$ of a function is small, then there exists a small coalition, biasing a function towards $1$.
\begin{lemma}
\label{lem: small 0 influence implies small coalition bias 1}
For all constant $\eps \in (0, 1)$, there exists a constant $C = C(\eps) \ge 1$ such that the following holds.
For all $n, r \in \N$, $\tau > 0$ and a monotone function $f: [r]^n \to \zo$ such that $\E_{x\sim \U_n}[f(x)] \ge \eps$ and $\J^0(f) \le \tau$, there exists $T\subseteq [n]$ with $|T| \le \exp(C\cdot \tau^2)$ so that
\[
    \P_T^1(f) \ge 1 - \eps.
\]
\end{lemma}

We prove this key lemma in \cref{subsec: high inf or small junta other direction}. We first show how it implies  \cref{thm: corrupt finite alphabet}.

\begin{proof}[Proof of \cref{thm: corrupt finite alphabet}]
We apply \cref{claim:monotonization} to assume without loss of generality that $f$ is monotone. We will first show how to find such $b$ and $B$, and then analyze and show that it has the claimed properties. 

\paragraph{Finding such $b$ and $B$}
Let $\delta = \frac{1}{100\cdot C_{\cref{lem: small 0 influence implies small coalition bias 1}}}\cdot \frac{\sqrt{\log(n)}}{n}$ where $C_{\cref{lem: small 0 influence implies small coalition bias 1}} \ge 1$ is the constant in \cref{lem: small 0 influence implies small coalition bias 1} for the parameter $\eps$.

We consider the following algorithm to find such a $B$.
\begin{itemize}
\item 
Let $Z_0 = \emptyset$ and let $f_0 = f$.

\item
For $i\ge 0$, proceed as follows.

\begin{itemize}
\item 
If $\Pr[f_i(x) = 0]\ge 1-\eps$, then let $b = 0$, $B = Z_i$ and terminate.

\item 
Otherwise, if there exists $j\in [n]\setminus Z_i$ such that $\J_j^0(f_i) \ge \delta$, then define $Z_{i+1} = Z_i\cup \{j\}$, and $f_{i+1}(x): [r]^{n - i - 1}\to \zo$ as the function obtained from fixing variable $j$ in $f_i$ to $1$, and proceed to the next iteration.

\item
Otherwise, if for all $j\in [n]\setminus Z_i$, $\J_j^0(f_i) < \delta$, then set $b = 1$, and apply \cref{lem: small 0 influence implies small coalition bias 1} (with $\tau = \delta n$) to $f_i$ to obtain the desired coalition $B$.
\end{itemize}
\end{itemize}

\paragraph{Analysis}
We now analyze the above algorithm.

As $\J^0_j(f_i) \ge \delta$, we easily see that $\Pr[f_{i+1} = 0] \ge \Pr[f_i = 0] + \delta$.
This shows that our algorithm must terminate in finite steps.
If our algorithm terminates by setting $b = 0$, then we will have that $\abs{B}\le \frac{1}{\delta} = O\left(\frac{n}{\sqrt{\log(n)}}\right)$ as desired.

Otherwise, our algorithm terminates by setting $b = 1$ in iteration $i$ for some $i\ge 0$. 
In that case, we have that $\abs{B}\le \exp(C_{\cref{lem: small 0 influence implies small coalition bias 1}}\cdot (\delta\cdot n)^2) = o(n / \sqrt{\log(n)})$, by our choice of constants in $\delta$.
To argue about correctness, by monotonicity of $f$, we see that
\[
\P_B^1(f) = \Pr_x[f(x_{-B}, r^B)=1] \ge \Pr_x[f(x_{-B-Z_i},r^B,1^i)=1] \ge 1 - \eps.
\]
\end{proof}

\subsection{Proving the direction switching lemma}
\label{subsec: high inf or small junta other direction}

In this subsection, we will prove \cref{lem: small 0 influence implies small coalition bias 1}, our result that if total influence of a function towards $0$ is small, then a small coalition can bias it towards $1$.

To prove this, we will need the following two ingredients.
First, we will need the following result of Hatami regarding functions with small $p$-biased influence.

\begin{theorem}[Simplified version of Corollary 2.10 from \cite{hatami12structure}]
\label{thm: hatami pseudo junta}
For all constants $\eps \in (0, 1)$, there exists a constant $C = C(\eps) \ge 1$ such that the following holds.
For all $n\in \N, 0 < p \le \frac{1}{2}$, and a monotone function $f: \zo^n \to \zo$ with $\E_{x\sim \mu_p} [f(x)] \ge \eps$, there exists $S\subseteq [n]$ with $|S| \le \exp\left(C\cdot \left(\I^p(f)\right)^2\right)$ such that 
\[
    \P_S^1(f, \mu_p) \ge 1 - \eps.
\]
\end{theorem}

Second, we need the following result that relates $p$-biased influence of a function to bias towards $0$ of the original function.
\begin{lemma}
\label{lem: approx product p biased}
For $r, n\in \N$, let $f: [r]^n \to \zo$ be a monotone function.
Then, for all $\eps> 0$, $p \in (0,1/2)$ there exists $m\in \N$ and a function $g: \zo^m \to\zo$ with the following properties.
\begin{enumerate}
\item 
$g$ is monotone.

\item
$\abs{\E_{y\sim \mu_p}[g(y)] - \E_{x\sim \U_n}[f(x)]} \le prn$.

\item
$\I^p(g) \le 2\cdot \J^0(f) + 4pr\cdot n^2$.

\item
For all $S\subseteq [m]$, if
\[
    \P_S^1(g, \mu_p) \ge 1 - \eps,
\]
then there exists $T\subseteq [n]$ with $|T|\le |S|$ such that
\[
    \P_T^1(f) \ge 1 - \eps - prn.
\]
\end{enumerate}
\end{lemma}

We will prove \cref{lem: approx product p biased} in \cref{sec:biased_space}.
We now show how we can combine these two ingredients and prove our direction switching lemma.

\begin{proof}
[Proof of \cref{lem: small 0 influence implies small coalition bias 1}]
We let 
\[
    p = \min\left(\frac{\eps}{2rn}, \frac{\tau}{2r\cdot n^2}\right)
\]
so that $prn \le \frac{\eps}{2}$ and $4pr\cdot n^2 \le 2\tau$.
With this, we apply \cref{lem: approx product p biased} to $f$ with parameters $\eps_{\cref{lem: approx product p biased}} = \frac{\eps}{2}$ to obtain a function $g: \zo^m \to \zo$ for some $m\in \N$ with the following properties:
\begin{enumerate}
\item 
$g$ is monotone.

\item
$\E_{y\sim \mu_p}[g(y)] \ge \eps - prn \ge \frac{\eps}{2}$.

\item
$\I^p(g) \le 2\tau + 4pr\cdot n^2 \le 4\tau$.

\item
For all $S\subseteq [m]$, if $\P_S^1(g, \mu_p) \ge 1 - \frac{\eps}{2}$, then there exists $T\subseteq [n]$ with $|T|\le |S|$ such that
\[
    \P^1_T(f) \ge 1 - \eps.
\]
\end{enumerate}

With this, we apply \cref{thm: hatami pseudo junta} with parameter $\eps_{\cref{thm: hatami pseudo junta}} = \frac{\eps}{2}$ to $g$ to find a $C = C(\eps) \ge 1$ and a coalition $S\subseteq [m]$ with $\abs{S} \le \exp(C\cdot \tau^2)$ such that
$\P_S^1(g, \mu_p) \ge 1 - \frac{\eps}{2}$.
By property 4 of $g$ above, the desired claim follows.
\end{proof}

\section{From functions in product space to \texorpdfstring{$p$-biased}{p-biased} space}\label{sec:biased_space}

We use this section to prove \cref{lem: approx product p biased}.
The main idea is to encode each coordinate so that every nonconstant
threshold of its encoded value is an OR of some of its bits.
Since fixing all other coordinates of a monotone function leaves a
threshold or a constant function, this will let us bound the
$p$-biased influence of the encoded function by the influence of $f$
towards $0$. We first prove this influence bound, then choose the
encoding to approximate the uniform distribution, and finally
complete the proof of \cref{lem: approx product p biased} at the end of \cref{subsec: helper claims}.
Throughout this section, fix a monotone function $f:[r]^n\to\zo$
and $p\in(0,1/2)$.

\subsection{The encoding and its influence bound}
\label{subsec:relating_influence_zero}

We assign each bit a value in $[r]$ and output the largest value
assigned to a bit that equals $1$, or $1$ if all bits are $0$.
Formally, we define the following max selector function.

\begin{definition}
For $\ell\in \N$ and a nondecreasing sequence $a = (a_1, \dots, a_{\ell}) \in [r]^{\ell}$, we define $\phi_a:\zo^{\ell} \to [r]$ as 
\[
    \phi_a(y) =
    \begin{cases}
    a_{\max\{i\in[\ell]:y_i=1\}} & y\ne 0^{\ell} \\
    1 & y = 0^{\ell}
    \end{cases}
\]
\end{definition}

The useful property is that, for every $t>1$,
\[
    \phi_a(y)\ge t
    \quad\Longleftrightarrow\quad
    \bigvee_{i\in[\ell]:\,a_i\ge t} y_i=1.
\]
Thus, every nonconstant threshold on $[r]$ becomes an OR of some
of the encoding bits. This property holds for every choice of $a$.
We will choose $a$ to approximate the uniform distribution in
\cref{subsec: exists seq s.t. phi approx p biased}.

For now, fix any nondecreasing sequence $a\in[r]^\ell$.
Set $m=\ell\cdot n$ and define $g:\zo^m\to\zo$ as
\[
    g(y)=f(\phi_a(y_1,\dots,y_\ell),\dots,
    \phi_a(y_{m-\ell+1},\dots,y_m)).
    \numberthis\label{eq: def g}
\]
Let $\A$ be the distribution on $[r]$ obtained by sampling
$y\sim\mu_p^\ell$ and outputting $\phi_a(y)$.
The encoded blocks are independent, so their joint distribution
is $\A^n$.

We now prove the  influence bound.

\begin{claim} We have
\label{claim: Inf p g at most 2 inf f A}
\[
    \I^p(g) \le 2\cdot \J^0(f, \A).
\]
\end{claim}

\begin{proof}
For $j\in[n]$, let $B_j=\{(j-1)\cdot\ell+1,\dots,j\cdot\ell\}$.
It suffices to show that
\[
    \sum_{i\in B_j}\I_i^p(g)\le 2\cdot\J_j^0(f,\A)
\]
for every $j\in[n]$.
Fix such a $j$ and condition on all blocks except $B_j$.
Let $x_{-j}$ denote their encoded values, and let
$h(y_{B_j})=f(\phi_a(y_{B_j}),x_{-j})$ be the remaining function,
where the encoded value is placed in coordinate $j$.

If fixing $x_{-j}$ makes $f$ constant, then $h$ is constant as well,
and both its total $p$-biased influence and the conditional
influence of coordinate $j$ towards $0$ are zero.

Otherwise, by monotonicity, there is a value $t^*>1$ such that
$f(z,x_{-j})=1$ if and only if $z\ge t^*$.
In particular, $f(1,x_{-j})=0$.
The function $h$ is therefore an OR of
$k=|\{i\in[\ell]:a_i\ge t^*\}|$ bits.
Each of these $k$ bits changes the output 
exactly when the other $k-1$ bits are zero and the two independent samplings of the bit (from $\mu_p$)
give different values. Consequently,
\begin{align*}
    \I^p(h)
    &=k\cdot2p(1-p)\cdot(1-p)^{k-1}\\
    &=2(1-p)\Pr[\text{exactly one of these $k$ bits is $1$}]\\
    &\le 2\Pr[h=1] \\
    &=2\Pr_{x_j\sim\A}[f(x_j,x_{-j})=1] \\
    & = 2(\Pr_{x_j\sim\A}[f(x_j,x_{-j})=1] - f(1,x_{-j})),
\end{align*}
where the final equality use the fact that $x^{-j}$ does not make $f$ a constant, and so $f(1,x_{-j})$ must be $0$.

Recalling that  $x_{-j}$ follows the distribution $\A^{n-1}$, we get
\[
    \sum_{i\in B_j}\I_i^p(g)
    \le 2\E_{x_{-j}\sim\A^{n-1}}
    \big[\Pr_{x_j\sim\A}[f(x_j,x_{-j})=1]-f(1,x_{-j})\big]
    =2\J_j^0(f,\A).
 \]
The claim now follows by summing over $j\in[n]$.
\end{proof}

\subsection{Approximating the uniform distribution}
\label{subsec: exists seq s.t. phi approx p biased}

We next choose $a$ so that $\A$ approximates the uniform
distribution on $[r]$. It is helpful to note that
\[
    \Pr_{y\sim\mu_p^\ell}[\phi_a(y)\le i]
    =(1-p)^{|\{j\in[\ell]:a_j>i\}|}.
\]
This just follows by definition of $\phi_a$ since the encoded value is at most $i$ exactly when every bit
assigned a value greater than $i$ is zero.
We choose the number of such bits to make this probability close
to $i/r$. Taking  differences for consecutive values of $i$ gives the following.

\begin{lemma}
\label{lem: exists seq s.t. phi approx p biased}
There exists $\ell\in \N$ and a nondecreasing sequence $a = (a_1, \dots, a_{\ell}) \in [r]^{\ell}$ such that for all $i\in [r]$, we have that
\[
    \abs{\Pr_{y\sim \mu_p}[\phi_a(y) = i] - \frac{1}{r}} \le 2p.
\]
\end{lemma}

\begin{proof}
For $i\in [r]$, let $\chi(i)$ be the unique nonnegative integer such that
\[
(1-p)^{\chi(i)+1} < \frac{i}{r} \le (1-p)^{\chi(i)}. \numberthis \label{eq: chi def}
\]
Let $a$ be the nondecreasing sequence containing one copy of $1$ and, for each $i\ge 2$, $\chi(i-1)-\chi(i)$ copies of $i$.
We observe that $\chi(r) = 0$ and so, for any $i\in [r]$, the number of terms in $a$ that have value $> i$ is $\chi(i)$.
Therefore, $\Pr_{y\sim \mu_p}[\phi_a(y) \le i] = (1-p)^{\chi(i)}$.

Thus, for $i \ge 2$, we have
\begin{align*}
\Pr_{y\sim \mu_p}[\phi_a(y) = i]
& = \Pr[\phi_a(y)\le i] - \Pr[\phi_a(y)\le i-1]\\
& = (1 - p)^{\chi(i)} - (1-p)^{\chi(i-1)} \numberthis  \label{eq: phi equals i}.
\end{align*}

To bound \cref{eq: phi equals i}, we will make use of the following claim, which is proved at the end of this subsection.
\begin{claim}
\label{claim: 1-p chi(i) tight bound} The following holds:
\[
   \frac{i}{r}\le (1 - p)^{\chi(i)} \le \frac{i}{r}(1+2p).
\]    
\end{claim}
First, consider the case of $i=1$.
Since $\phi_a$ takes values in $[r]$, we have
\[
    \Pr_{y\sim\mu_p}[\phi_a(y)=1]
    = (1-p)^{\chi(1)}.
\]
Applying \cref{claim: 1-p chi(i) tight bound}, we obtain
\[
    \frac1r
    \le \Pr_{y\sim\mu_p}[\phi_a(y)=1]
    \le \frac1r(1+2p)
    \le \frac1r+2p,
\]
yielding the desired bound for $i=1$.

For the remainder of the proof, fix $i\ge 2$. First, using \cref{eq: phi equals i}, we compute that
\begin{align*}
\Pr_{y\sim \mu_p}[\phi_a(y) = i] 
& = (1 - p)^{\chi(i)} - (1-p)^{\chi(i-1)}\\
& \le \frac{i}{r}(1+2p) - \frac{i-1}{r}\\
& \le \frac{1}{r} + 2p
\end{align*}
where for the middle inequality, we applied \cref{claim: 1-p chi(i) tight bound}.

Similarly, again using \cref{eq: phi equals i}, we have that
\begin{align*}
\Pr_{y\sim \mu_p}[\phi_a(y) = i] 
& = (1 - p)^{\chi(i)} - (1-p)^{\chi(i-1)}\\
& \ge \frac{i}{r} - \frac{i-1}{r}(1+2p)\\
& \ge \frac{1}{r} - 2p
\end{align*}
where for the middle inequality, we again applied \cref{claim: 1-p chi(i) tight bound}.

Hence, we conclude that
\[
\abs{\Pr_{y\sim \mu_p}[\phi_a(y) = i] - \frac{1}{r}} \le 2p
\]
as desired.
Lastly, we prove our remaining claim.
\begin{proof}[Proof of \cref{claim: 1-p chi(i) tight bound}]
By definition of $\chi(i)$ from \cref{eq: chi def}, our lower bound of $(1-p)^{\chi(i)} \ge \frac{i}{r}$ follows.

Similarly, from \cref{eq: chi def}, we compute our upper bound as
\begin{align*}
(1-p)^{\chi(i)}
& \le \frac{i}{r}\cdot \frac{1}{1-p}\\
& = \frac{i}{r}\cdot \left(1 + \frac{p}{1-p}\right)\\
& \le \frac{i}{r}\cdot \left(1 + 2p\right),
\end{align*}
where the last inequality uses the fact that $p \in (0,1/2)$.
\end{proof}
\end{proof}

\subsection{Completing the proof}
\label{subsec: helper claims}

Fix a sequence $a$ that satisfies the properties in \cref{lem: exists seq s.t. phi approx p biased}.
We first record the resulting bounds on the distance from uniform.

\begin{lemma}
\label{lem: A is close to uniform}
The following holds regarding $\A$.
\begin{enumerate}
\item 
\[
\norm{\A - \U^r}_{\infty} \le 2p.
\]
\item
For all $i\ge 1$,
\[
\abs{\A^i - \U^r_i} \le pri.
\]
\end{enumerate}
\end{lemma}

\begin{proof}
[Proof of \cref{lem: A is close to uniform}]
The first part is immediate from \cref{lem: exists seq s.t. phi approx p biased}.
For the second part, we proceed by induction on $i$.
For $i=1$, by definition of statistical distance and the first part,
\[
    \abs{\A-\U^r}
    = \frac12\sum_{a\in[r]}
      \abs{\Pr[\A=a]-\frac1r}
    \le pr.
\]
For $i\ge 2$, applying the triangle inequality, we obtain
\begin{align*}
    \abs{\A^i-\U^r_i}
    &\le \abs{\A^{i-1}\times\A-\U^r_{i-1}\times\A}
       + \abs{\U^r_{i-1}\times\A-\U^r_{i-1}\times\U^r}\\
    &= \abs{\A^{i-1}-\U^r_{i-1}}+\abs{\A-\U^r}\\
    &\le pr(i-1)+pr
     = pri,
\end{align*}
where the equality holds because taking a product with the same
distribution preserves statistical distance, and the last inequality
follows from the induction hypothesis and the case $i=1$.
\end{proof}

We next relate the influence of $f$ towards $0$ under $\A$ to its influence under the uniform distribution.

\begin{claim}
\label{claim: Inf f uniform vs A} It holds that
\[
    \J^0(f, \A)\le \J^0(f) + 2pr\cdot n^2.
\]
\end{claim}

\begin{proof}[Proof of \cref{claim: Inf f uniform vs A}]
By monotonicity of $f$,
\(
\P_i^0(f,\A)
= \Pr_{x\sim\A^n}[f(1,x_{-i})=0]\), and the same equation holds under the uniform distribution as well.
Applying \cref{claim: data processing expectation} and
\cref{lem: A is close to uniform} to the indicators of
$f(1,x_{-i})=0$ and $f(x)=0$, respectively, we obtain
\[
\P_i^0(f,\A)\le \P_i^0(f)+prn
\qquad\text{and}\qquad
\Pr_{x\sim\A^n}[f(x)=0]
\ge \Pr_{x\sim\U_n}[f(x)=0]-prn.
\]
Therefore,
\begin{align*}
\J^0(f,\A)
&= \sum_i \left(
\P_i^0(f,\A)-\Pr_{x\sim\A^n}[f(x)=0]
\right)\\
&\le \sum_i \left(
\P_i^0(f)+prn
-\left(\Pr_{x\sim\U_n}[f(x)=0]-prn\right)
\right)\\
&= \J^0(f)+2pr\cdot n^2,
\end{align*}
as desired.
\end{proof}

With the above bounds, we are ready to prove \cref{lem: approx product p biased}.

\begin{proof}
[Proof of \cref{lem: approx product p biased}]
We prove each of these results one after the other.
\begin{enumerate}
\item 
Let $y_1\le y_2\in \zo^m$, and let $x_1,x_2\in[r]^n$
be their corresponding encoded values.
By definition of $\phi_a$, changing an input bit from $0$ to $1$
cannot decrease the encoded value of its block.
Thus, $x_1\le x_2$. Since $f$ is monotone, we obtain
$g(y_1)=f(x_1)\le f(x_2)=g(y_2)$, as desired.

\item
By definition of $g$, we have that $\E_{y\sim \mu_p}[g(y)] = \E_{x\sim \A^n}[f(x)]$.
Hence, it suffices to show that
\[
\abs{\E_{x\sim \A^n}[f(x)] - \E_{x\sim \U_n}[f(x)]}\le prn.
\]
Applying the data processing inequality (\cref{claim: data processing expectation}), it suffices to show that $\abs{\A^n - \U_n} \le prn$.
This last claim directly follows from \cref{lem: A is close to uniform}.

\item
This directly follows from \cref{claim: Inf f uniform vs A} and \cref{claim: Inf p g at most 2 inf f A}.

\item
Let $S\subseteq [m]$ be arbitrary such that $\P^1_S(g,\mu_p)\ge 1-\eps$.
For $i\in [n]$, let $B_i = \{(i-1) \cdot \ell+1, \dots, i \cdot \ell\}$.
We define our set $T$ as follows:
$T = \{i\in [n]: B_i \cap S \ne \emptyset\}$.
Then, $\abs{T}\le \abs{S}$.
We claim that
\[
\P^1_T(f, \A)
= \Pr_{x\sim \A^n}[f(r^T,x_{-T})=1]
\ge 1-\eps.
\]
First, let us see how our desired claim follows from this.
From \cref{lem: A is close to uniform}, we have that
$\abs{\A^n-\U_r^n}\le prn$.
Applying \cref{claim: data processing expectation} to the function
obtained by fixing the coordinates in $T$ to $r$, we have that
\[
\P^1_T(f)
= \Pr_{x\sim \U_r^n}[f(r^T,x_{-T})=1]
\ge \P^1_T(f,\A)-prn
\ge 1-\eps-prn
\]
as desired.

We now prove our remaining claim that $\P^1_T(f,\A)\ge 1-\eps$.
Let $S'=\bigcup_{i\in T} B_i$.
Setting all bits in each block indexed by $T$ to $1$ produces
encoded values at most $r$.
By monotonicity of $f$, fixing the original coordinates in $T$
to $r$ can only increase the output. Thus,
$\P^1_T(f,\A)\ge \P^1_{S'}(g,\mu_p)$.
By the definitions of $T$ and $S'$, we have $S\subseteq S'$.
Therefore,
\[
\P^1_T(f,\A)
\ge \P^1_{S'}(g,\mu_p)
\ge \P^1_S(g,\mu_p)
\ge 1-\eps.
\]
\end{enumerate}
\end{proof}

\subsection*{Conclusion}

After this work was made public, \cite{Friedgut26junta} improved our bounds to show that a coalition of $O_{\eps}(n / \log(n))$ coordinates can bias any function over a large alphabet to some target output. Many natural questions remain open such as whether this bound is optimal. 

Another related question is to study an analogous notion of resampling influence defined over the product space.
This was undertaken in \cite{Hatami09DTProduct, keller11product, hatami12structure}.\footnote{See the latest arXiv version of \cite{Hatami09DTProduct} at \url{https://arxiv.org/abs/math/0612405v3}.}  
Even though we use the notion of resampling influence on $p$-biased space to prove our results, our techniques do not yield structure results for resampling influence over the product space.

Another direction to explore is regarding the study of $p$-biased influence itself.
Starting with the work of Margulis \cite{margulis1974probabilistic} and Russo \cite{russo1981critical}, and later expanded upon by Friedgut (with appendix by Bourgain) \cite{FB99sharp}, many rich connections to extremal combinatorics, threshold phenomena, statistical physics and more are known. This has led to an intense study of influence in $p$-biased spaces \cite{talagrand1994russo, FK96every, friedgut98junta, bourgain00harmonic, KK07thresholds, hatami12structure}.
One can easily embed a boolean function over the $p$-biased cube as a Boolean function over the solid cube and our result shows a reverse connection can hold as well under suitable encoding. 
While we relied on the pseudo-junta theorem of Hatami \cite{hatami12structure}, the subsequent result of Friedgut \cite{Friedgut26junta} uses different arguments to improve the bound to $\frac{n}{\log n}$.
It is natural to ask if using these techniques of Friedgut \cite{Friedgut26junta}, one can obtain a pseudo-junta theorem for $p$-biased influence with size dependence $\exp(O_{\eps}(\I^p(f)))$, improving upon the result of \cite{hatami12structure}. Such a bound, albeit with weaker guarantees than a pseudo-junta theorem, was recently obtained \cite{KLLM24hypercontractivity}.

\paragraph{Acknowledgements}
M.G. would like to thank David Zuckerman for helpful conversations.
Part of the work was done while the authors were visiting the Simons Institute for Theory of Computing, UC Berkeley. We would like to thank the Simons Institute for the stimulating research environment and the hospitality.

\paragraph{AI use statement.}
The authors used ChatGPT Pro for exploring ideas, and proving several technical parts in \cref{sec:biased_space}.
The authors take full responsibility for the correctness of all results in the paper.

\printbibliography

\end{document}